\documentclass[11pt]{article}
\usepackage[abs]{overpic}
\usepackage[hmarginratio={1:1},     
vmarginratio={1:1},     
textwidth=17cm,        
textheight=24cm,
heightrounded,]
{geometry}
\usepackage{tocloft}
\usepackage{enumitem}
\usepackage{pdfpages}
\usepackage{graphicx}
\usepackage{amsmath,accents}
\usepackage{amsthm}
\usepackage{amssymb}
\usepackage{mathrsfs, mathtools}
\usepackage[mathscr]{euscript}
\mathtoolsset{showonlyrefs}

\usepackage{tikz}
\usepackage{textcomp}
\usepackage{pdfpages}
\usepackage{tikz-cd}
\usepackage{tikz-3dplot}
\usepackage{pgfplots}
\usetikzlibrary{calc}
\pgfplotsset{width=7cm,compat=1.8}
\usetikzlibrary{decorations.pathreplacing}
\usetikzlibrary{decorations.markings}
\usetikzlibrary{patterns}
\usepgflibrary{shapes.geometric}

\tikzset{
	string/.style={draw=#1, postaction={decorate}, decoration={markings,mark=at position .51 with {\arrow[color=#1]{>}}}},
	costring/.style={draw=#1, postaction={decorate}, decoration={markings,mark=at position .51 with {\arrow[draw=#1]{<}}}},
	ostring/.style={draw=#1, postaction={decorate}, decoration={markings,mark=at position .47 with {\arrow[draw=#1]{>}}}},
	ustring/.style={draw=#1, postaction={decorate}, decoration={markings,mark=at position .56 with {\arrow[draw=#1]{>}}}},
	oostring/.style={draw=#1, postaction={decorate}, decoration={markings,mark=at position .43 with {\arrow[draw=#1]{>}}}},
	uustring/.style={draw=#1, postaction={decorate}, decoration={markings,mark=at position .59 with {\arrow[draw=#1]{>}}}},
	directed/.style={string=blue!50!black}, 
	odirected/.style={ostring=blue!50!black}, 
	udirected/.style={ustring=blue!50!black}, 
	oodirected/.style={oostring=blue!50!black}, 
	uudirected/.style={uustring=blue!50!black},     
	redirected/.style={costring= blue!50!black},
	redirectedgreen/.style={costring= green!50!black},
	directedgreen/.style={string= green!50!black},
	redirectedlightgreen/.style={costring= green!65!black},
	directedlightgreen/.style={string= green!65!black},
	redirectedred/.style={costring= red!50!black},
	directedred/.style={string= red!50!black}%
}

\tikzset{-dot-/.style={decoration={
			markings,
			mark=at position 0.5 with {\fill circle (1.875pt);}},postaction={decorate}}}

\tikzset{
	Fdot/.style={circle, draw, fill, inner sep=0pt}, 
	Odot/.style={circle, draw, inner sep=0.1pt, minimum size=0.1cm}
}

\usepackage{import}
\usepackage{calc}

\newcommand\tikzzbox[1]
{#1}

\makeatletter
\DeclareRobustCommand{\cev}[1]{%
	{\mathpalette\do@cev{#1}}%
}
\newcommand{\do@cev}[2]{%
	\vbox{\offinterlineskip
		\sbox\z@{$\m@th#1 x$}%
		\ialign{##\cr
			\hidewidth\reflectbox{$\m@th#1\vec{}\mkern4mu$}\hidewidth\cr
			\noalign{\kern-\ht\z@}
			$\m@th#1#2$\cr
		}%
	}%
}
\makeatother

\let\emph\undefined
\newcommand{\emph}[1]{\textsl{#1}}

\usepackage[hidelinks]{hyperref}

\numberwithin{equation}{section}
\newtheoremstyle{style1}
{13pt}
{13pt}
{}
{}
{\normalfont\bfseries}
{.}
{.5em}
{}
\theoremstyle{style1}

\newtheorem{definition}[equation]{Definition}
\newtheorem{example}[equation]{Example}
\newtheorem{remark}[equation]{Remark}

\newtheoremstyle{style2}
{13pt}
{13pt}
{\slshape}
{}
{\normalfont\bfseries}
{.}
{.5em}
{}

\theoremstyle{style2}

\newtheorem{theorem}[equation]{Theorem}
\newtheorem{proposition}[equation]{Proposition}

\newtheorem{conjecture}[equation]{Conjecture}

\newtheorem{expectation}[equation]{Expectation}
\newtheorem{proposal}[equation]{Proposal}

\usepackage{dsfont}
\usepackage{tikz}
\usetikzlibrary{matrix,arrows,decorations.pathmorphing}
\usepackage{tikz-cd}

\usepackage{tikz}
\usetikzlibrary{matrix,arrows,decorations.pathmorphing,shapes.geometric}

\usepackage{wasysym}

\usepackage{tikz}

\definecolor{Blue} {rgb} {0.282352,0.239215,0.803921}
\definecolor{Green} {rgb} {0.133333,0.545098,0.133333}
\definecolor{Red}   {rgb} {0.803921,0.000000,0.000000}
\definecolor{Violet}{rgb} {0.580392,0.000000,0.827450}

\newcounter{jfc}

\newcommand{\R}{\mathbb{R}}

\newcommand{\Z}{\mathbb{Z}}

\newcommand{\Ca}{\mathcal{C}}

\newcommand{\Vect}{{\mathsf{Vect}}}

\newcommand{\ev}{{\mathrm{ev}}}
\newcommand{\coev}{{\mathrm{coev}}}

\newcommand{\Cat}{\operatorname{\mathscr{C}at}}

\newcommand{\End}{\mathsf{End}}
\newcommand{\Hom}{\mathsf{Hom}}

\newcommand{\id}{\text{id}}

\newcommand{\Bord}{\mathrm{Bord}}

\let\to\undefined
\newcommand{\to}{\longrightarrow}

\newcommand{\AdjCat}{{\mathsf{AdjCat}}}
\newcommand{\RigidCat}{{\mathsf{RigidCat}}}

\newcommand{\op}{\operatorname{op}}

\DeclareMathSymbol{\Phiit}{\mathalpha}{letters}{"08} 
\DeclareMathSymbol{\Psiit}{\mathalpha}{letters}{"09}
\DeclareMathSymbol{\Sigmait}{\mathalpha}{letters}{"06}
\DeclareMathSymbol{\Xiit}{\mathalpha}{letters}{"04}
\DeclareMathSymbol{\Piit}{\mathalpha}{letters}{"05}\let\Pi\undefined\newcommand{\Pi}{\Piit}
\DeclareMathSymbol{\Gammait}{\mathalpha}{letters}{"00}
\DeclareMathSymbol{\Omegait}{\mathalpha}{letters}{"0A}
\DeclareMathSymbol{\Upsilonit}{\mathalpha}{letters}{"07}
\DeclareMathSymbol{\Thetait}{\mathalpha}{letters}{"02}

\title{On the Higher Categorical Structure of Topological Defects in Quantum Field Theories}
\author{Lukas Müller}

\begin{document}
\maketitle
\begin{abstract}
\noindent 
We propose a unifying mathematical framework describing the higher categorical structures formed by topological defects in quantum field theory equipped with tangential structures, such as orientations, framings, or $\operatorname{Pin}^{\pm}$-structures, in terms of structured versions of higher dagger categories. This recovers all previously known results, including the description of oriented topological defects in 2-dimensional quantum field theories by pivotal bicategories. 
Assuming the stratified cobordism hypothesis, we prove our proposal for topological defects with stable tangential structures that admit a direct sum in fully extended topological quantum field theories. 
\end{abstract}
\tableofcontents
\section{Introduction}
Topological defects are extended observables of a quantum field theory localized on submanifolds of spacetime whose value does not change under deformations of the submanifold they are supported on. Mathematically, their configurations are described by stratified manifolds. 
Over the last few years topological defects have received tremendous attention, since they can be interpreted as generalized or higher categorical symmetries, see~\cite{snow,GS2,GS1, FMT} for recent reviews. This has led to many exciting applications in physics and new developments in mathematics. 

Traditional symmetries can be used to construct codimension 1 topological defects such that the two sides of the defect differ by the symmetry. The idea behind generalized or categorical symmetries is that even if we drop the condition on invertibility and codimension 1 topological defects still retain many of the important features of symmetries, such as anomalies, conserved charges, gauging and can be used to, for example, constrain RG-flows.    

It is believed, and proven in low dimensions, that topological defects in quantum field theories assemble into a
higher category $\mathcal{D}$, which encodes many of their physical properties; see, for example,~\cite{defects} for a recent review. Quantum field theories are the objects of $\mathcal{D}$. 1-morphisms represent topological codimension 1 defects between these theories, with their composition given by fusion—the operation of bringing defects close together. Similarly, 2-morphisms correspond to codimension 2 topological defects between codimension 1 defects, and this pattern continues with higher morphisms representing higher codimensional defects. In practice, it is often too hard to determine the full category of \emph{all} topological defects and one has only access to simpler subcategories of $\mathcal{D}$. 
Usually these categories of topological defects are linear. We will mostly ignore this point in the present paper, because it has been discussed in detail before, see for example~\cite{kapustin2010topological,johnson2022classification} and focus on a less explored aspect of their structure. 

The dependence on tangential structures (see~\cite{debray2025global} for a discussion of the global structure), such as orientations or spin structures, for higher categories of topological defects remain an important open problem. This paper proposes a systematic mathematical framework to address it. Defects of different codimensions may each carry distinct tangential structures, and their interactions can be highly non-trivial. We use the approach of Ayala, Francis and Rozenblyum~\cite{AFR,AFR19} to define tangential structures on defects using the tangent bundle of conically smooth stratified spaces.     

So far the additional structure present on $\mathcal{D}$ depending on the tangential structure has only been studied for oriented defects in dimension two and three: In~\cite{dkr1107.0495, CMS} it is shown that oriented defects in oriented theories in dimensions two and three are encoded by pivotal bicategories and Gray categories with duals, respectively.

Let us explain the structure in dimension two in slightly more detail:
Every 1-morphism $X\colon \mathcal{Z}_1\to \mathcal{Z}_2$, i.e.\ a 1-dimensional topological defect between two quantum field theories, has a \textsl{left} and \textsl{right adjoints}, meaning there are 1-morphisms ${}^\vee\!X \colon \mathcal{Z}_2\to \mathcal{Z}_1$ and $X^\vee \colon \mathcal{Z}_2\to \mathcal{Z}_1$ both given by reversing the orientation of $X$ together with 2-morphisms 
\begin{equation} 
	\label{eq:AdjMaps}
	\ev_X  = 
	\tikzzbox{%
		\begin{tikzpicture}[very thick,scale=1.0,color=blue!50!black, baseline=0.5cm]
			\coordinate (X) at (0.5,0);
			\coordinate (Xd) at (-0.5,0);
			\coordinate (d1) at (-1,0);
			\coordinate (d2) at (+1,0);
			\coordinate (u1) at (-1,1.25);
			\coordinate (u2) at (+1,1.25);
			%
			\fill [orange!40!white, opacity=0.7] (d1) -- (d2) -- (u2) -- (u1); 
			\draw[thin] (d1) -- (d2) -- (u2) -- (u1) -- (d1); 
			%
			\draw[directed] (X) .. controls +(0,1) and +(0,1) .. (Xd);
			%
			\fill (X) circle (0pt) node[below] {{\small $X\vphantom{X^\vee}$}};
			\fill (Xd) circle (0pt) node[below] {{\small ${}^\vee\!X$}};
			\fill[red!80!black] (0,0.3) circle (0pt) node {{\scriptsize $\mathcal{Z}_1$}};
			\fill[red!80!black] (0.8,1) circle (0pt) node {{\scriptsize $\mathcal{Z}_2$}};
		\end{tikzpicture}
	}
	\, , \quad 
	\coev_X = 
	\tikzzbox{%
		\begin{tikzpicture}[very thick,scale=1.0,color=blue!50!black, baseline=0.5cm]
			\coordinate (X) at (-0.5,1.25);
			\coordinate (Xd) at (0.5,1.25);
			\coordinate (d1) at (-1,0);
			\coordinate (d2) at (+1,0);
			\coordinate (u1) at (-1,1.25);
			\coordinate (u2) at (+1,1.25);
			%
			\fill [orange!40!white, opacity=0.7] (d1) -- (d2) -- (u2) -- (u1); 
			\draw[thin] (d1) -- (d2) -- (u2) -- (u1) -- (d1);
			%
			\draw[redirected] (X) .. controls +(0,-1) and +(0,-1) .. (Xd);
			%
			\fill (X) circle (0pt) node[above] {{\small $X\vphantom{X^\vee}$}};
			\fill (Xd) circle (0pt) node[above] {{\small ${}^\vee\!X$}};
			\fill[red!80!black] (0,0.9) circle (0pt) node {{\scriptsize $\mathcal{Z}_2$}};
			\fill[red!80!black] (0.8,0.2) circle (0pt) node {{\scriptsize $\mathcal{Z}_1$}};
		\end{tikzpicture}
	}
	\quad\textrm{and}\quad 
	\tilde{\ev}_X  = 
	\tikzzbox{%
		\begin{tikzpicture}[very thick,scale=1.0,color=blue!50!black, baseline=0.5cm]
			\coordinate (X) at (0.5,0);
			\coordinate (Xd) at (-0.5,0);
			\coordinate (d1) at (-1,0);
			\coordinate (d2) at (+1,0);
			\coordinate (u1) at (-1,1.25);
			\coordinate (u2) at (+1,1.25);
			%
			\fill [orange!40!white, opacity=0.7] (d1) -- (d2) -- (u2) -- (u1); 
			\draw[thin] (d1) -- (d2) -- (u2) -- (u1) -- (d1);
			%
			\draw[redirected] (X) .. controls +(0,1) and +(0,1) .. (Xd);
			%
			\fill (Xd) circle (0pt) node[below] {{\small $X\vphantom{X^\vee}$}};
			\fill (X) circle (0pt) node[below] {{\small $X^\vee$}};
			\fill[red!80!black] (0,0.3) circle (0pt) node {{\scriptsize $\mathcal{Z}_2$}};
			\fill[red!80!black] (0.8,1) circle (0pt) node {{\scriptsize $\mathcal{Z}_1$}};
		\end{tikzpicture}
	}
	\, , \quad 
	\tilde{\coev}_X = 
	\tikzzbox{%
		\begin{tikzpicture}[very thick,scale=1.0,color=blue!50!black, baseline=0.5cm]
			\coordinate (X) at (-0.5,1.25);
			\coordinate (Xd) at (0.5,1.25);
			\coordinate (d1) at (-1,0);
			\coordinate (d2) at (+1,0);
			\coordinate (u1) at (-1,1.25);
			\coordinate (u2) at (+1,1.25);
			%
			\fill [orange!40!white, opacity=0.7] (d1) -- (d2) -- (u2) -- (u1); 
			\draw[thin] (d1) -- (d2) -- (u2) -- (u1) -- (d1);
			%
			\draw[directed] (X) .. controls +(0,-1) and +(0,-1) .. (Xd);
			%
			\fill (Xd) circle (0pt) node[above] {{\small $X\vphantom{X^\vee}$}};
			\fill (X) circle (0pt) node[above] {{\small $X^\vee$}};
			\fill[red!80!black] (0,0.9) circle (0pt) node {{\scriptsize $\mathcal{Z}_1$}};
			\fill[red!80!black] (0.8,0.2) circle (0pt) node {{\scriptsize $\mathcal{Z}_2$}};
		\end{tikzpicture}
	} 
	\, , 
\end{equation}
respectively. These satisfy the \emph{snake identities}
\begin{equation}
	\label{eq:ZorroMoves}
	\tikzzbox{%
		\begin{tikzpicture}[very thick,scale=1.0,color=blue!50!black, baseline=0cm]
			\coordinate (A) at (-1,1.25);
			\coordinate (A2) at (1,-1.25);
			\coordinate (d1) at (-1.5,-1.25);
			\coordinate (d2) at (+1.5,-1.25);
			\coordinate (u1) at (-1.5,1.25);
			\coordinate (u2) at (+1.5,1.25);
			%
			\fill [orange!40!white, opacity=0.7] (d1) -- (d2) -- (u2) -- (u1); 
			\draw[thin] (d1) -- (d2) -- (u2) -- (u1) -- (d1);
			%
			\draw[directed] (0,0) .. controls +(0,-1) and +(0,-1) .. (-1,0);
			\draw[directed] (1,0) .. controls +(0,1) and +(0,1) .. (0,0);
			\draw (-1,0) -- (A); 
			\draw (1,0) -- (A2); 
			%
			\fill ($(A)+(0.1,0)$) circle (0pt) node[below left] {{\small $X$}};
			\fill ($(A2)+(-0.1,0)$) circle (0pt) node[above right] {{\small $X$}};
			\fill[red!80!black] (-0.5,0) circle (0pt) node {{\scriptsize $\mathcal{Z}_2$}};
			\fill[red!80!black] (+0.5,0) circle (0pt) node {{\scriptsize $\mathcal{Z}_1$}};
		\end{tikzpicture}
	}
	= 
	\tikzzbox{%
		\begin{tikzpicture}[very thick,scale=1.0,color=blue!50!black, baseline=0cm]
			\coordinate (A) at (0,1.25);
			\coordinate (A2) at (0,-1.25);
			\coordinate (d1) at (-1,-1.25);
			\coordinate (d2) at (+1,-1.25);
			\coordinate (u1) at (-1,1.25);
			\coordinate (u2) at (+1,1.25);
			%
			\fill [orange!40!white, opacity=0.7] (d1) -- (d2) -- (u2) -- (u1); 
			\draw[thin] (d1) -- (d2) -- (u2) -- (u1) -- (d1);
			%
			\draw (A) -- (A2); 
			%
			\fill ($(A2)+(-0.1,0)$) circle (0pt) node[above right] {{\small $X$}};
			\fill[red!80!black] (-0.5,0) circle (0pt) node {{\scriptsize $\mathcal{Z}_1$}};
			\fill[red!80!black] (+0.5,0) circle (0pt) node {{\scriptsize $\mathcal{Z}_2$}};
		\end{tikzpicture}
	}
	\, , \quad 
	\tikzzbox{%
		\begin{tikzpicture}[very thick,scale=1.0,color=blue!50!black, baseline=0cm]
			\coordinate (A) at (1,1.25);
			\coordinate (A2) at (-1,-1.25);
			\coordinate (d1) at (-1.5,-1.25);
			\coordinate (d2) at (+1.5,-1.25);
			\coordinate (u1) at (-1.5,1.25);
			\coordinate (u2) at (+1.5,1.25);
			%
			\fill [orange!40!white, opacity=0.7] (d1) -- (d2) -- (u2) -- (u1); 
			\draw[thin] (d1) -- (d2) -- (u2) -- (u1) -- (d1);
			%
			\draw[directed] (0,0) .. controls +(0,1) and +(0,1) .. (-1,0);
			\draw[directed] (1,0) .. controls +(0,-1) and +(0,-1) .. (0,0);
			\draw (-1,0) -- (A2); 
			\draw (1,0) -- (A); 
			%
			\fill ($(A)+(-0.1,0)$) circle (0pt) node[below right] {{\small ${}^\vee X$}};
			\fill ($(A2)+(0.1,0)$) circle (0pt) node[above left] {{\small ${}^\vee X$}}; 
			\fill[red!80!black] (-0.5,0) circle (0pt) node {{\scriptsize $\mathcal{Z}_1$}};
			\fill[red!80!black] (+0.5,0) circle (0pt) node {{\scriptsize $\mathcal{Z}_2$}};
		\end{tikzpicture}
	}
	= 
	\tikzzbox{%
		\begin{tikzpicture}[very thick,scale=1.0,color=blue!50!black, baseline=0cm]
			\coordinate (A) at (0,1.25);
			\coordinate (A2) at (0,-1.25);
			\coordinate (d1) at (-1,-1.25);
			\coordinate (d2) at (+1,-1.25);
			\coordinate (u1) at (-1,1.25);
			\coordinate (u2) at (+1,1.25);
			%
			\fill [orange!40!white, opacity=0.7] (d1) -- (d2) -- (u2) -- (u1); 
			\draw[thin] (d1) -- (d2) -- (u2) -- (u1) -- (d1);
			%
			\draw (A) -- (A2); 
			%
			\fill ($(A2)+(-0.1,0)$) circle (0pt) node[above right] {{\small ${{}^\vee X}$}};
			\fill[red!80!black] (-0.5,0) circle (0pt) node {{\scriptsize $\mathcal{Z}_2$}};
			\fill[red!80!black] (+0.5,0) circle (0pt) node {{\scriptsize $\mathcal{Z}_1$}};
		\end{tikzpicture}
	}
\end{equation} 
for the left adjoint, and analogous conditions for the right adjoint. 

The defining property of a pivotal bicategory is that left and right adjoints agree, which is specific to oriented defects. 
There is an interesting and fruitful interplay between categorical structures and the type of allowed defect: The `bending' of strata, allows to form diagrams in a pivotal bicategory~\cite{cr1210.6363} and in 3-dimensions in a Gray category with duals~\cite{BMS} which cannot be evaluated in ordinary higher categories such as
\begin{equation}\label{Page9}
	\tikzzbox{
		\begin{tikzpicture}[very thick,scale=0.7,color=blue!50!black, baseline=-0.1cm]
			\fill [orange!40!white, opacity=0.7] (-1.7,-1.7) -- (-1.7,1.7) -- (1.7,1.7) -- (1.7,-1.7); 
			\draw[thin] (-1.7,-1.7) -- (-1.7,1.7) -- (1.7,1.7) -- (1.7,-1.7) -- (-1.7,-1.7);
			\fill (-135:1.32) circle (0pt) node[red!80!black] {{\scriptsize$\mathcal{Z}_2$}};
			\fill (-135:0.55) circle (0pt) node[red!80!black] {{\scriptsize$\mathcal{Z}_1$}};
			\draw (0,0) circle (0.95);
			\fill (45:1.32) circle (0pt) node {{\small$X$}};
			\draw[<-, very thick] (0.100,-0.95) -- (-0.101,-0.95) node[above] {}; 
			\draw[<-, very thick] (-0.100,0.95) -- (0.101,0.95) node[below] {}; 
			\fill (0:0.95) circle (3.5pt) node[left] {{\small$\chi$}};
		\end{tikzpicture} 
	} 
	\quad \quad  \quad  \quad \text{or} \quad \quad \quad \quad \quad 	
	\tikzzbox{\begin{tikzpicture}[thick,scale=1.961,color=blue!50!black, baseline=0.8cm, >=stealth, 
			style={x={(-0.6cm,-0.4cm)},y={(1cm,-0.2cm)},z={(0cm,0.9cm)}}]
			\coordinate (Lb) at (0, 0, 0);
			\coordinate (Rb) at (0, 1.5, 0);
			\coordinate (Rt) at (0, 1.5, 1);
			\coordinate (Lt) at (0, 0, 1);
			\coordinate (d) at (0, 0.45, 0.5);
			\coordinate (b) at (0, 1.05, 0.5);
			\fill [orange!80,opacity=0.545] (Lb) -- (Rb) -- (Rt) -- (Lt);
			\fill[color=red!80!black] (d) circle (0pt) node[left] (0up) { {\scriptsize$X$} };
			\fill[inner color=green!30!white,outer color=green!55!white, very thick, rounded corners=0.5mm] (d) .. controls +(0,0,0.4) and +(0,0,0.4) .. (b) -- (b) .. controls +(0,0,-0.4) and +(0,0,-0.4) .. (d);
			\draw[color=red!80!black, ultra thick, rounded corners=0.5mm, postaction={decorate}, decoration={markings,mark=at position .51 with {\arrow[draw=red!80!black]{>}}}] (d) .. controls +(0,0,0.4) and +(0,0,0.4) .. (b);
			\draw[color=red!80!black, ultra thick, rounded corners=0.5mm, postaction={decorate}, decoration={markings,mark=at position .53 with {\arrow[draw=red!80!black]{<}}}] (d) .. controls +(0,0,-0.4) and +(0,0,-0.4) .. (b);
			\fill[color=red!80!black] (0, 1.35, 0.15) circle (0pt) node (0up) { {\scriptsize$u$} };
			\fill[color=green!50!black] (0, 0.85, 0.4) circle (0pt) node (0up) { {\scriptsize$w$} };
			\fill[color=red!80!black, opacity=0.2] (0, 0.65, 0.6) circle (0pt) node (0up) { {\scriptsize$u$} };
	\end{tikzpicture}}
	\ \ ,
\end{equation}     
respectively. In these pictures all strata are oriented in a compatible way. 

Another instructive and simple example is the structure of topological defects in unoriented 1-dimensional quantum field theories. We can assign to every 1-morphism $f\colon \mathcal{Z}_1\to \mathcal{Z}_2$ corresponding to a local defect configuration of the form 
\begin{center}
\begin{overpic}[scale=1]{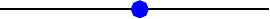} 
\put(35,10){{$\mathcal{Z}_1$}}
\put(90,10){{$\mathcal{Z}_2$}}
\put(65,13){$f$}
\end{overpic}
\end{center}
another defect $f^\dagger\colon \mathcal{Z}_2\to \mathcal{Z}_1$ by reading the picture right to left, which is impossible for oriented theories. The source of $f$ and $f^\dagger$ differ and hence it does not make sense for them to be equal. The operation $(-)^\dagger$ reverses the order of composition, i.e.\ $(f\circ g)^\dagger= g^\dagger \circ f^\dagger$ and squares to the identity: $(f^{\dagger})^\dagger=f$. We conclude that the category of topological defects for unoriented 1-dimensional theories is equipped with a functor $\dagger\colon \mathcal{D}\to \mathcal{D}^{\op}$ which is the identity on objects and squares to the identity. A category with such a functor is called a dagger category~\cite{DefDag}. 

Often this additional $\dagger$-structure is the only way of distinguishing oriented and unoriented defect categories. For example for topological field theories with target $\Vect$ the oriented 
defect category agrees with the category of finite dimensional vector spaces and linear maps whereas the unoriented defect category corresponds to vector spaces equipped with a symmetric non-degenerate pairing and arbitrary linear maps between them. 
As categories, these two categories are equivalent because every finite dimensional vector space admits a non-degenerate symmetric pairing. The distinction is that the latter is equipped with a natural dagger operation given by the adjoint for the pairings. 

Both examples for the additional structures present on categories of topological defects have the common feature that they involve conditions which are unnatural from a categorical perspective, namely they require objects or 1-morphisms to be equal instead of isomorphic: $(-)^\dagger$ is the identity on objects and the left and right adjoint agree. We believe that this is a feature rather than a bug of the construction and any general mathematical framework to describe defects should have it. 

In this paper, we propose an abstract mathematical framework systematically encoding the dependence on arbitrary types of tangential structure, which should be closely related to the envisioned notion of an $H$-pivotal category by Kevin Walker~\cite{walker2021universal}. Concretely, we define structured version of higher dagger categories~\cite{higherdagger} depending on stratified tangential structures. All examples from this introduction turn out to be special cases of the general definition. 

 \vspace*{0.2cm}\textsc{Acknowledgments.} 
I am grateful to my collaborators Nils Carqueville, Cameron Krulewski, Dave Penneys, Claudia Scheimbauer, Pelle Steffens, Luuk Stehouwer, and David Reutter for the enjoyable collaborations which inspired this paper and, especially, to Theo Johnson-Freyd for invaluable discussions on technical and conceptional aspects of it.      
I thank Severin Bunk and Lukas Woike
for helpful discussions related to this project and comments on an early draft. 
I gratefully acknowledge support of the Simons Collaboration on Global Categorical Symmetries. Research at Perimeter Institute is supported in part by the Government of Canada through the Department of Innovation, Science and Economic Development and by the Province of Ontario through the Ministry of Colleges and Universities.

\section{Tangential structures for stratified manifolds}\label{Sec: Tangential structures}
The geometry of topological defects is best encoded using stratified manifolds~\cite{AFT} where every stratum carries a label indicating the specifics of the defect located on it. 
In this section, we will review how to equip stratified manifolds with tangential structures following~\cite{AFR}.  

Throughout this paper, a stratified manifold will always be a conically smooth manifold as defined in~\cite{AFT}. Roughly, these are inductively defined as stratified topological spaces together with a (maximal) covering by opens of the form $\R^j\times C(Y)$, where $C(Y)$ is the cone over a conically smooth stratified manifold of lower depth, with smooth transition functions. Here are a few examples:
\begin{center}
\includegraphics[scale=1.2]{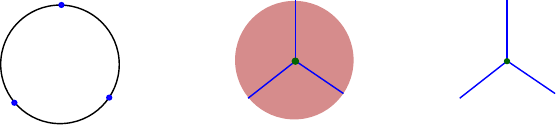}
\end{center}
The second stratified manifold is the cone on the first and the last is constructed by forgetting the top strata of the second. 

To any stratified topological space $Y$ we can associate its $(\infty,1)$ exit path category~\cite[Appendix A.5]{ha} whose objects are the points of $Y$ and whose morphisms are path which are allowed to exit lower strata to enter higher strata, but are not allowed to return. For example the exit path category of a codimension 1 plane in $\R^n$ is equivalent to the walking span 
\begin{equation}
    \begin{tikzcd}
        & * \ar[ld] \ar[rd] & \\
        * & & * 
    \end{tikzcd} \ \ . 
\end{equation}
The core $\iota_0 (\operatorname{Exit}(Y))$ of the exit path category is the disjoint union of the $\infty$-groupoids associated to the strata of $Y$. In~\cite{AFR19} Ayala, Francis and Rozenblyum show that conically smooth stratified manifolds have a tangent bundle classified by a functor $TY\colon \operatorname{Exit}(Y)\to \mathsf{Vect}_{\operatorname{inj}}$ where $\mathsf{Vect}_{\operatorname{inj}}$ is the $(\infty,1)$-category represented by the topological category of vector spaces and injections. Restricting the stratified tangent bundle to the core of the exit path category recovers the individual tangent bundles of the strata. However, the stratified tangent bundle also knows about the way those fit together. For example, if the stratification is simply an embedding of a submanifold $\Sigma \hookrightarrow M $, then the value on a path leaving $\Sigma$ is the combination of the embedding $T\Sigma \hookrightarrow TM$ with the value of the path in $TM$.     

To answer the question what are tangential structures for stratified manifolds, we  
recall that a space $X\to BO_n$ specifies a tangential structure on $n$-dimensional manifolds in terms of lifts 
\begin{equation}
\begin{tikzcd}
    & X \ar[d] \\ 
    M\ar[r,"TM",swap] \ar[ru] & BO_n
\end{tikzcd}
\end{equation}
of the map classifying the tangent bundle to $X$. 
Using the stratified tangent bundle we can define tangential structures as lifts against a conservative functor $\mathcal{B}\to \mathsf{Vect}_{\operatorname{inj}}$ from a $(\infty,1)$-category $\mathcal{B}$. 
\begin{example}\label{Ex: 1strata}
We look at the simplified situation describing  $k$-dimensional defects in $n$-dimensional spacetimes (see also~\cite{tetik2022stratified} for a detailed discussion of this example). For this, we denote by $\Vect_{\operatorname{inf}}^{k,n}$ the full subcategory of $\Vect_{\operatorname{inj}}$ on vector spaces of dimension $k$ and $n$. The stratified tangent bundles of a straification only containing $k$ and $n$ dimensional strata takes values in this category. This category has up to equivalence two objects with endomorphisms $O_k$ and $O_n$. 
The only other non-trivial hom space is $\Hom(\R^k,\R^n) \simeq O_n/O_{n-k}$. We look at categories $\mathcal{B}$ having the same shape, i.e. two objects whose endomorphisms from groups $G_k$ and $G_n$, respectively, as well as one space of non-trivial morphisms $X'$ which has an action of $G_k$ and $G_n$ by pre and post composition. A functor from this category to $\Vect_{\operatorname{inf}}^{k,n}$ is the same as group homomorphisms $G_k\to O_k$ and $G_n\to O_n$ as well as an equivariant map $X'\to  O_n/O_{n-k}$. This data can be specified by a map of spans
\begin{equation}
\begin{tikzcd}
  & X'/\hspace{-0.1cm}/ (G_k\times G_n) \ar[ld] \ar[rd] \ar[dd] & \\
  BG_k \ar[dd] & & BG_n \ar[dd] \\ 
  & (O_n/O_{n-k})/\hspace{-0.1cm}/(O_k\times O_n)\cong BO_k\times BO_{n-k} \ar[ld,"\operatorname{pr}",swap] \ar[rd,"\oplus"] & \\
  BO_k &  & BO_n
\end{tikzcd} \ \ , 
\end{equation} 
where $/\hspace{-0.1cm}/$ denotes the homotopy quotient. 
 Explicitly, a tangential structure of this type for a $k$-dimensional manifold $\Sigma$ embedded into an $n$-dimensional manifold $M$ consists of\footnote{I am grateful to David Reutter for explaining these structures to me long before I understood the general formalism.} 
\begin{itemize}
\item a $G_n$-structure on $M\setminus \Sigma$, 
\item a $G_k$-structure on $\Sigma$,
\item $\Sigma$ does not only have a tangent bundle but also a normal bundle in $M$ leading to a map $\Sigma \to BO_k\times BO_{n-k}$ combining both. We need to specify a homotopy lift of this map to $X'/\hspace{-0.1cm}/ (G_k\times G_n)$. 
\item The previous choice defines a $G_n$-structure on a tubular neighborhood $\mathcal{U}$ of $\Sigma$. We need to specify a homotopy between this tangential structure and the restriction of the $G_n$ structure on $M\setminus \Sigma$ to $\mathcal{U}$. It also defines a $G_k$-structure on $\Sigma$ and we need to specify a homotopy between it and the $G_k$-structure fixed above. 
\end{itemize}
Often it is natural to consider tangential structures where the structure on the tubular neighborhood fixes the structure on $\Sigma$ uniquely.

As a concrete example of a stratified tangential structure we take $G_n=G_k=X'=*$ leading to the notion of a vari framing introduced by Ayala-Francis-Rozenblyum~\cite{AFR}, which consists of a framing on $M$, a tangential and normal framing of $\Sigma$, and a homotopy between the sum of the normal and tangential framing and the ambient framing of $M$. 
\end{example}
Note that without loss of generality we can assume that all endomorphisms in $\mathcal{B}$ are isomorphisms which is ensured by requiring the functor $\mathcal{B}\to \Vect_{\operatorname{inj}}$ to be conservative. For us the following definition will be important. 
\begin{definition} Let $n$ be a positive integer.
A \emph{stratifed $n$-dimensional tangential structure} is a conservative functor $\mathcal{B}_n\to \Vect_{\operatorname{inj}}^{\leq n}$, where $\Vect_{\operatorname{inj}}^{\leq n}$ is the category of at most $n$-dimensional vector spaces and injections.     
\end{definition}
We can give a description of this structure in terms similar to Example~\ref{Ex: 1strata}. 
\begin{remark}\label{Rem: X}
A stratified $n$-dimensional tangential structure consists of 
\begin{itemize}
    \item For all $0\leq i\leq n$ a space $X_i$ over $BO_i$ 
    \item For all $0\leq i < k \leq n $ a map of spans of spaces 
    
    \begin{equation}
\begin{tikzcd}
  & X_{i,k} \ar[ld] \ar[rd] \ar[d] & \\
  X_k \ar[d] & BO_k\times BO_{i-k} \ar[ld,"\operatorname{pr}",swap] \ar[rd,"\oplus"] & X_i \ar[d]  \\
  BO_k &  & BO_i
\end{tikzcd}
\end{equation}
\item For all $0\leq l < i < k \leq n $ maps of maps of spans $X_{l,i}\times_{X_i} X_{i,k} \to X_{l,k}$ 
    \item and so one up to an $(n-1)$-morphism filling a diagram of various ways to relate $X_{0,1}\times_{X_1} X_{1,2}\times_{X_2} \dots \times_{X_{n-1}} X_{(n-1),n}$ with $X_{0,n}$. 
\end{itemize}
\end{remark}

\begin{example}\label{Ex: Hn}
To every stable tangential structure $BH\to BO$ we can assign a stratified $n$-dimensional tangential structure called $\Vect_{\operatorname{inj}}^{H_n}$. Its objects are $k$-dimensional vector spaces $V$ with an $H_k$-structure, i.e. an $H_k$-torsor $P$ equipped with an equivariant map to the orthogonal frames on $V$, for $k\leq n$. 
Endomorphisms are maps of $H_k$-structures and maps between vector spaces of different dimensions are just injections. The corresponding tangential structure equips all $k$-strata with an $H_k$ structure without imposing any compatibility relations between different strata. The tangential structures $\Vect_{\operatorname{inj}}^{SO_n}$ are the most common ones in the context of topological defects.     

If $H$ admits a direct sum operation the corresponding maps of spans are 
 \begin{equation}
\begin{tikzcd}
  & BH_i\times BH_{i-k} \ar[ld] \ar[rd, "\oplus"] \ar[d] & \\
  BH_k \ar[d] & BO_k\times BO_{i-k} \ar[ld,"\operatorname{pr}",swap] \ar[rd,"\oplus"] & BH_i \ar[d]  \\
  BO_k &  & BO_i
\end{tikzcd} \ \ . 
\end{equation}
\end{example}

\begin{example}
Ayala, Francis, and Rozenblyum introduce two different notions of framings~\cite{AFR}: 
\begin{itemize}
    \item Vari framings corresponding to 
    \begin{align}
       \mathbb{N}^{\leq n} \coloneqq \{ 0\rightarrow 1 \rightarrow\dots \rightarrow n \} & \longrightarrow \Vect_{\operatorname{inj}}^{\leq n} \\ 
        i & \longmapsto \R^i
    \end{align} which maps the non-trivial map $i\longrightarrow i+k$ to the embedding of $\R^i $ into $\R^{i+k}$ as the first $i$-coordinates. In particular, for a vari-framed stratified manifold  every $i$-stratum is equipped with an $i$-framing. 
    \item Solid framings corresponding to the functor $\Vect_{\operatorname{inj}}/\R^n \to \Vect_{\operatorname{inj}}^{\leq n}$. Every vari-framing is also a solid framing, but many stratified manifolds admit a solid framing but no vari-framing, such as a circle embedded into $\R^2$. 
\end{itemize}
\end{example}

A \emph{pointed} stratified tangential structure is a tangential structure $\mathcal{B}_n$ together with the choice of compatible basepoints for all the spaces $X_i$, and $X_{ij}$ from Remark~\ref{Rem: X} which are preserved by all maps. A pointing can be reformulated as a map of tangential structures from vari-framings
\begin{equation}
    \begin{tikzcd}
        \mathbb{N}^{\leq n} \ar[rr]\ar[rd] & & \mathcal{B}_n \ar[ld] \\ 
         & \Vect_{\operatorname{inj}}^{\leq n} &  \ \ . 
    \end{tikzcd}
\end{equation}

\section{Structured higher dagger categories}
Let $\mathcal{B}_n$ be a stratifed $n$-dimensional tangential structure.
In this section we define the notions of $\mathcal{B}_n$-dagger categories and rigid symmetric monoidal (rigid for short) $\mathcal{B}_n$-dagger categories. 
For $n>2$, the first definition relies on an unproven conjecture. However, the definition of rigid dagger categories only relies on the cobordism hypothesis and even the case $n=2$ leads to many novel concepts and examples. 
Recall that an $n$-category $\Ca$ has \emph{adjoints} if all its $(<n)$-morphism have a left and right adjoint. 
We start by reformulating the definition of $O_n$-dagger categories from~\cite[Definition 5.7 and 5.8]{higherdagger} in a way which straightforwardly generalizes to arbitrary (pointed) stratified tangential structures. The construction is based on the following conjecture regarding the structure of $n$-categories with adjoints:
\begin{conjecture}\label{Conj}
There is a functor $\AdjCat_{\bullet} \colon \Vect_{\operatorname{inj}}^{\operatorname{op}} \longrightarrow \Cat $ sending $\R^n$ to the $(\infty,1)$-category $\AdjCat_n$ of $(\infty,n)$-categories with adjoints and the inclusion $\R^i\to \R^n$ as the first $i$ coordinates to the functor $\iota_i\colon \AdjCat_n \to \AdjCat_i$ which sends an $n$-category with adjoints to the $i$-category with adjoints constructed by forgetting all non-invertible $(>i)$-morphisms.
\end{conjecture}
One reason to believe this conjecture is that $n$-categories with adjoints are expected to admit a graphical calculus in terms of (vari) framed higher string diagrams in $\R^n$. 
These should be functorial for injections between vector spaces by extending a higher string diagram as constant in the orthogonal direction.\footnote{In joint work in progress~\cite{MSS}, we will make this more precise.} 
Another justification is that Conjecture~\ref{Conj} is a theorem when restricted $\Vect_{\operatorname{inj}}^{\leq 2}$ and $(2,2)$-categories~\cite{2DuCH}.

A particular consequence of the conjecture is that $O_n$ acts on $\AdjCat_n$ (see also~\cite[Remark 4.4.10]{luriecobhyp}) allowing us to define \emph{$O_n$-volutive $n$-categories} as fixed points for this action. Being functorial in $\Vect_{\operatorname{inj}}$ implies that $\iota_i\colon \AdjCat_n \to \AdjCat_i$ is $O_i\times O_{n-i}$-equivariant where $O_{n-i}$ acts trivially on $\AdjCat_{i}$. 

The notion of a lax limit allows for a concise definition of $O_n$-dagger categories. To explain this in more detail, we recall first that an object of the limit of a functor $F\colon D \to \Cat $ consists of a collection of objects $c_d\in F(d)$ together with isomorphisms $F(f)[c_d]\cong c_{d'}$ for all 1-morphisms $f\colon d \to d' $, together with higher coherence isomorphisms. Formally the limit can be realized as the category of natural transformation from the constant diagram at the terminal category $*$. The lax limit of $F$ replaces the isomorphisms $F(f)[c_d]\cong c_{d'}$ with a not necessarily invertible morphism $c_{d'}\to F(f)[c_d]$, which can be realized by the category of oplax-natural transformations from the constant diagram at $*$~\cite{gepner2017lax, johnson2017op}.

Let $\AdjCat_{\bullet}^{\leq n}$ denote the restriction of $\AdjCat_{\bullet}$ to $\Vect_{\operatorname{inj}}^{\leq n}$. The lax limit $\lim^{\operatorname{Lax}} \AdjCat_{\bullet}^{\leq n}$ of this functor is an $(\infty,1)$-category with objects consisting of
\begin{itemize}
    \item An $O_i$-equivariant functor $*\to \AdjCat_i$, which can be identified with an $O_i$-volutive category $\Ca_i$ for all $0 \leq i \leq n$ together with 
    \item $O_i\times O_{k-i}$-volutive functors $\Ca_i\to \iota_i \Ca_k$ using the trivial $O_{k-i}$-volution on $\Ca_i$ and
    \item various coherence isomorphisms between these functors.
\end{itemize}
This reproduces the structure which is part of the definition of an $O_n$-dagger category, but we are missing some additional fully-faithfulness and essential surjectivity conditions. 
\begin{remark}
The fully-faithfulness condition requires some preparation (the reader can safely ignore these details):
Consider the diagram consisting of the  $(\infty,n-l)$-categories $(\iota_{n-l} \mathcal{C}_{n-j})^{O_{k_1} \times \dots \times O_{k_m}}$, where $(k_1, \dots, k_m)$ is a non-empty partition of $j$ for $j \leq l$.
There are two types of maps in the diagram: those forgetting fixed point data corresponding to $(k_1, \dots, k_i+k_{i+1} ,\dots , k_m) \rightsquigarrow (k_1, \dots, k_{i}, k_{i+1}, \dots, k_m)$ and those relating different $\Ca_i$'s corresponding to $(k_1, \dots, k_{m-1}) \rightsquigarrow (k_1, \dots, k_{m-1}, k_m)$.
We denote by $P_{n-l}(\mathcal{C})$ the pullback of this diagram. There is a canonical map $\mathcal{C}_{n-l} \to P_{n-l}(\mathcal{C})$. 
For example, for $n=3$ and $l=3$ the diagram is:
\begin{equation}\label{Eq: Faithful}
\begin{tikzcd}[sep=small]
 & \mathcal{C}_0\\
	&& P_0(\mathcal{C}) && {\iota_0\mathcal{C}_2^{O_2}} \\
	{\iota_0\mathcal{C}_3^{O_3}} &&& {\iota_0\mathcal{C}_3^{O_2\times O_1}} \\
	&& {\iota_0\mathcal{C}_1^{O_1}} && {\iota_0\mathcal{C}_2^{O_1 \times O_1}} \\
	{\iota_0\mathcal{C}_3^{O_1 \times O_2}} &&& {\iota_0\mathcal{C}_3^{O_1 \times O_1 \times O_1}}
	\arrow[from=3-1, to=5-1]
	\arrow[from=5-1, to=5-4]
	\arrow[from=4-3, to=5-1]
	\arrow[from=4-3, to=4-5]
	\arrow[from=4-5, to=5-4]
	\arrow[from=2-5, to=4-5]
	\arrow[from=2-5, to=3-4]
	\arrow[from=2-3, to=3-1]
	\arrow[from=2-3, to=2-5]
	\arrow[from=2-3, to=4-3]
	\arrow[from=3-1, to=3-4,crossing over]
	\arrow[from=3-4, to=5-4, crossing over]
        \arrow[from=1-2, to=2-3, dashed]
\end{tikzcd}
\end{equation}
\end{remark}
After these preparations we can reformulate the definition from~\cite{higherdagger}.
\begin{definition}\label{Def: O dagger}
An \emph{$O_n$-dagger category} is an element of $\lim^{\operatorname{Lax}} \AdjCat_{\bullet}^{\leq n}$ such that the functors $\Ca_i\to \iota_i\Ca_k$ are essentially surjective on $(\leq i)$-morphisms and the maps $\mathcal{C}_{i} \to P_{i}(\mathcal{C})$ are fully faithful on $(> i)$-morphisms. The category of $O_n$-dagger categories is the full subcategory of $\lim^{\operatorname{Lax}} \AdjCat_{\bullet}^{\leq n}$ on $O_n$-dagger categories.   
\end{definition}
\begin{remark}
Note that it does not make sense to ask for a trivialization of an $O_n$-volution, since it is a fixed point for a non-trivial action. However, this action trivializes on the core and hence we can trivialize it ``on objects". For $n=1$, these trivializations are called hermitian structures in~\cite[Appendix B]{freed2021reflection}.  
From this perspective the definition of $O_n$-dagger categories has a simple interpretation: An $O_n$-dagger category is an $O_n$-volutive category together with a trivialization of the volution on the space of objects, encoded by $\Ca_0$, a trivialization of the $O_{n-1}$-action on the space of 1-morphisms encoded by $\Ca_1$ and so on. 
The essential surjectivity conditions ensure that the $O_{n-i}$-volution is trivialized on \emph{all} $i$-morphisms, whereas the fully faithfulness conditions ensure that $\Ca_i$ does not contain any new information about $(>i)$-morphisms. 
\end{remark}

\begin{example}
For $n=1$, this recovers the traditional notion of a dagger category in its coherent reformulation introduced in~\cite{luukjan}. This is equivalent to the previously introduced strict definition of a dagger category as a category $\Ca$ equipped with a functor $\dagger\colon \Ca\to\Ca^{\op}$ which is the identity on objects such that $\dagger^2=\id$. Prominent examples are the category of Hilbert spaces and bordism categories~\cite{luukthesis}.
\end{example}
\begin{example}
For $n=2$, the resulting structure will be studied extensively in~\cite{2DuCH} (see also~\cite{TimNils}) where we show that examples include the bicategory of Baez super 2-Hilbert spaces and bordism bicategories. We conjecture, but do not prove, in~\cite{2DuCH} that a strict model for an $O_2$-dagger bicategory is a bicategory with adjoints $B$ equipped with 
\begin{enumerate}
    \item a functor $\dagger: B \to B^{2\op}$ such that $\dagger^2 = \id_B$ and $x^\dagger = x$ and $f^\dagger = f$ for objects and $1$-morphisms;
    \item a left adjoint functor $(.)^L: B \to B^{1\op, 2\op}$ which commutes with $\dagger$.
\end{enumerate}
\end{example}
An advantage of the reformulation of the definition of $O_n$-dagger categories in terms of a lax limit is that the generalization to other tangential structures is straightforward. 
\begin{definition}\label{Def: Bn dagger}
Let $\mathcal{B}_n \to \Vect_{\operatorname{inj}}^{\leq n}$ be a pointed stratifed $n$-dimensional tangential structure. A \emph{$\mathcal{B}_n$-dagger category} is an object of the lax limit $\lim_{\mathcal{B}_n}^{\operatorname{Lax}} \AdjCat_{\bullet}^{\leq n} $ satisfying appropriate fully-faithfulness and surjectivity conditions generalizing Definition~\ref{Def: O dagger} (explained in more detail below). The category of $\mathcal{B}_n$-dagger categories is the full subcategory of the lax limit on $\mathcal{B}_n$-dagger categories.
\end{definition}
The assumption that $\mathcal{B}_n$ is pointed is not necessary to define the lax limit, but seems necessary when we want to be able to impose surjectivity and faithfulness conditions. For example if $X_{i,i+1}$ is empty there is not even a functor we could impose such a condition on. 

Informally, a $\mathcal{B}_{n}$-dagger category consists of a family of $X_i$-volutive categories, i.e.\ fixed points $\Ca_i\in \AdjCat_i^{X_i}$, $X_{ij}$-volutive functors $\Ca_i\to \iota_i \Ca_k$ which are essentially surjective on $(\leq i)$-morphisms, with many additional coherence isomorphisms, and fully-faithfulness conditions.   
Let us explain the fully-faithfulness conditions in more detail.\footnote{This is already a subtle point for the definition of $O_n$-dagger categories and there is a chance that our definitions have to be revised in the future. Their goal is to imply that the larger than $k$-morphisms in $\mathcal{C}_k$ do not add any new information.} For $(k_1, \dots, k_m)$ a partition of $j$ for $j \leq l$ as above we denote by $F_{i,k_1,\dots, k_m}$ the fiber of the map \begin{align}
    X_{i,i+k_1}\bigtimes_{X_{i+k_1}} X_{i+k_1, i+k_1+k_2} \bigtimes_{X_{i+k_1+k_2}} \dots \bigtimes_{X_{i+\dots + k_{m-1}}} X_{i+\dots + k_{m-1}, i+\dots + k_m} \to X_i \ \ . 
\end{align}  
The $\mathcal{B}^n$ version of Equation~\ref{Eq: Faithful} is now constructed by replacing every appearance of $O_{k_1}\times \dots \times O_{k_m}$ with $F_{i,k_1,\dots, k_m}$.
\begin{example}
A vari-framed dagger $n$-category is the same as an $n$-category with adjoints, since the fully-faithfulness and surjectivity conditions imply that $\Ca_i \to \iota_i \Ca_n$ is an equivalence for all $i\leq n$. 

\end{example}

\begin{example}
We will call an $\Vect_{\operatorname{inj}}^{H_n}$-dagger category for the tangential structure from Example~\ref{Ex: Hn} an $H_n$-dagger category for short. 
The arguments in~\cite[Example 5.10]{higherdagger} can be made precise using the results of~\cite{2DuCH} implying that ${SO_2}$-dagger categories are equivalent to pivotal bicategories. A \emph{pivotal bicategory} is a bicategory $B$ with adjoints equipped with a natural isomorphism $\id_B\Longrightarrow (-)^{LL}$ whose components on all objects are the identity or equivalently coherent choices of morphism $X^\vee$ which are both left and right adjoints for all 1-morphisms $X$ as mentioned in the introduction (see also~\cite{TimNils}).

The notion of an $SO_3$-dagger categories should be closely related to the concept of Gray categories with duals~\cite{BMS}.
\end{example}
Another structure which is modeled on oriented topological defects in $\R^n$ is the notion of disk-like $n$-categories introduced by Morrison and Walker in their definition of blob-homology~\cite{morrison2011higher}. They prove that disk-like 2-categories are also equivalent to pivotal bicategories. These facts suggest the following:  
\begin{conjecture}
There is an equivalence between $SO_n$-dagger categories and disk-like $n$-categories.
\end{conjecture}
A problem with Definition~\ref{Def: O dagger} and~\ref{Def: Bn dagger} is that for $n>2$ they rely on the unproven Conjecture~\ref{Conj}. If all categories involved are in addition rigid symmetric monoidal we can translate them into well-defined quantities using the following conjecture from~\cite{higherdagger}. 
\begin{conjecture}\label{Conj2}
The $O_n$-action on $\RigidCat_{n}$ induced from the action in Conjecture~\ref{Conj} canonically trivializes. This induces an $O_{n-k}$-action on $\iota_k \Ca$ for every rigid $n$-category which agrees with the $O_{n-k}$-action induced by the cobordism hypothesis. 
\end{conjecture} 
The reasoning behind this conjecture is that a rigid symmetric monoidal category admits infinity many deloopings all of which have adjoints. 
These deloopings are expected to be compatible with the graphical calculus in terms of vari-framed string-diagrams. Hence, a rigid symmetric monoidal $n$-category should have a graphical calculus in terms of $n$-dimensional string diagrams in $\R^\infty$ on which $O_n$ acts by rotation. However, this space is contractible and hence the $O_n$-action trivializes. Again, for $n=2$ and $(2,2)$-categories Conjecture~\ref{Conj2} is a theorem~\cite{2DuCH}.

Hence we can replace for rigid category the ill-defined functor featuring in the definition. 
\begin{proposition}
There is a functor $\RigidCat_{\bullet}\colon \Vect_{\operatorname{inj}}^{\op}\to \Cat$ sending $\R^n$ to $\RigidCat_n$ equipped with the trivial $O_n$-action and the cobordism hypothesis action on morphisms. 
\end{proposition}
\begin{proof}
Let $\mathsf{Cat}_{(\infty,\infty)}^{\otimes}$ be the $\infty$-subcategory of categorical spectra~\cite{Naruki} on symmetric monoidal $(\infty,\infty)$-categories equipped with its natural tensor product $\otimes$ and the corresponding internal hom $$\underline{\operatorname{Hom}}\colon \mathsf{Cat}_{(\infty,\infty)}^{\otimes} \times \mathsf{Cat}_{(\infty,\infty)}^{\otimes} \to \mathsf{Cat}_{(\infty,\infty)}^{\otimes} \ \ . $$ We denote by $\Bord_n^*$ the framed fully extended $n$-dimensional bordism category, which is equipped with a canonical $O_n$-action. A consequence of the cobordism hypothesis is that it satisfies the following universal property
\begin{align}\label{Eq: Universal property}
\underline{\operatorname{Hom}}(\Bord_n^*; \Ca ) = \operatorname{Adj}_n(\Ca)
\end{align}
where $\operatorname{Adj}_n(\Ca)$ is the subcategory of $\Ca$ containing only $n$-dualizable objects and morphisms. Note that if $\Ca$ is an $n$-category then morphisms are $n$-dualizable if and only if they are invertible. The universal property implies that the functor 
\begin{align}
    \iota_k \colon \RigidCat_n \to \RigidCat_k 
\end{align}
is represented by $\underline{\operatorname{Hom}}(\Bord^*_{n-k};-)$. Hence the value of $\iota_k$ on a category is equipped with an $O_{n-k}$-action which we can combine with the trivial $O_k$-action to make $\iota_{k}$ into an $O_k\times O_{n-k}$-equivariant functor. The compatibility with composition follows from $$ \underline{\Hom} (\Bord_{k-i}^*; \underline{\Hom}(\Bord_{n-k}^*;-)) \cong \underline{\Hom}(\Bord_{k-i}^*\otimes \Bord_{n-k}^*;-) $$ and the $O_{k-i}\times O_{n-k}$-equivariant identification $\Bord_{k-i}^*\otimes \Bord_{n-k}^*\cong \Bord_{n-i}^* $ which follows from the universal property~\eqref{Eq: Universal property}.  
\end{proof}
Using this we can rigorously define. 
\begin{definition}
Let $\mathcal{B}_n \to \Vect_{\operatorname{inj}}^{\leq n}$ be a pointed stratifed $n$-dimensional tangential structure. A \emph{rigid $\mathcal{B}_n$-dagger category} is an object of the lax limit $\lim_{\mathcal{B}_n}^{\operatorname{Lax}} \RigidCat_{\bullet}^{\leq n} $ satisfying analogues of the fully-faithfulness and surjectivity conditions from Definition~\ref{Def: O dagger}. The category of rigid $\mathcal{B}_n$-dagger categories is the full subcategory of the lax limit on $\mathcal{B}_n$-dagger categories.
\end{definition}
Note that assuming all the conjectures above a rigid $\mathcal{B}_n$-dagger category has an underlying $\mathcal{B}_n$-dagger category.  
Concretely, a rigid $\mathcal{B}_n$-dagger category consists of 
\begin{itemize}
    \item rigid $i$-categories $\mathcal{C}_i$ equipped with $X_i$-actions for all $0\leq i \leq n$, 
    \item symmetric monoidal $X_{ij}$ equivariant functors $\mathcal{C}_i\to \iota_i \mathcal{C}_{i+k}$ where $X_{ij}$-acts on $\mathcal{C}_i$ through the map to $X_i$ and on $\iota_{i}\Ca_{j}$ by combining the action from the map to $X_{j}$ with the action from the map to $BO_k$. This is possible because the cobordism hypothesis action commutes canonically with any symmetric monoidal functor, 
    \item and coherence isomorphisms such that the essential surjectivity and fully-faithfulness conditions hold. 
\end{itemize}
\begin{example}
An $O_1$-rigid dagger category is the same as a rigid symmetric monoidal dagger category, i.e.\  a symmetric monoidal category equipped with a symmetric monoidal dagger functor $\dagger\colon \Ca\to \Ca^{\op}$ such that all coherence isomorphisms are unitary. The functor sending an object to its dual is not required to be a dagger-functor.
This requirement gets added by asking for a rigid $O_2$-dagger structure on the delooping $B\Ca$ as explained in more detail in~\cite[Example 5.11]{higherdagger}.   
\end{example}

\section{Categories of topological defects are higher dagger categories} 
Whatever a $n$-dimensional quantum field theory is, it should allow us to compute numbers associated to geometric objects, such as correlation functions of operators inserted on a spacetime manifold $M$ equipped with appropriate geometric structures, such as a metric, a principal bundle with connection, or  a Spin-structure. In the following we will denote the geometric structure by $\mathcal{F}$. 
To include defects, spacetime manifolds should be equipped with stratifications whose strata carry labels specifying the type of defect. A defect is called topological if all quantities we can compute are invariant under deformations of its position.  

For concreteness we fix a mathematical framework making this more precise known as functorial field theory, which defines QFTs as (smooth) symmetric monoidal functors $\mathcal{Z}\colon \Bord_n^\mathcal{F} \longrightarrow \mathcal{T}$~\cite{FMT,grady2021geometric,stolz2011supersymmetric, luriecobhyp, atiyah1988topological}. Here $\Bord_n^\mathcal{F}$ is a (smooth) symmetric monoidal higher category of bordisms encoding the cutting and gluing of spacetime manifolds along manifolds of arbitrary codimension and $\mathcal{T}$ is an appropriate target encoding the algebraic structures (complex numbers, state spaces, $\dots$) assigned to manifolds of different dimensions. We do not expect our exposition to depend on this choice of model. It is expected that one can include stratification into the definition of $\Bord_n^\mathcal{F}$ leading to a category $\Bord_n^{D,\mathcal{F}}$ whose morphisms roughly look like   
\begin{center}
\begin{overpic}[scale=0.5]{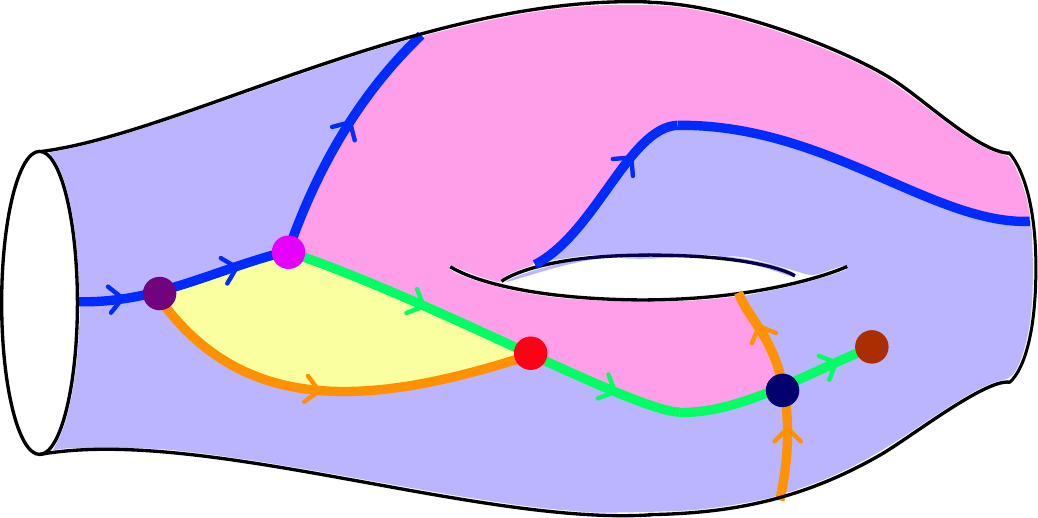}
\end{overpic}  .
\end{center}
To the best of our knowledge these higher categories have not been constructed rigorously even in the setting of topological field theories (for symmetric monoidal 1-categories with oriented defects definitions can be found in~\cite{carqueville2019orbifolds}). Quantum field theories with topological defects can be modeled as symmetric monoidal functors $\Bord_n^{D,\mathcal{F}} \to \mathcal{T}$ where we allow to vary the stratification in the source and require that they only depend on a stratified tangential structure. 

\begin{remark}
The notion of stratified tangential structure is too general for the physical intuition that topological defects form a higher category. For example consider the tangential structure corresponding to $\iota_0\Vect_{\operatorname{inj}}^{\leq n}\to \Vect_{\operatorname{inj}}^{\leq n}$. A stratified manifold with this tangential structure is just a disjoint union of manifolds of dimensions less or equal to $n$. There is no reason for a defect field theory which can evaluate these manifolds to form a higher category. To exclude these examples we will from now on restrict to pointed tangential structures, i.e. those equipped with the additional choice of a map 
\begin{equation}
\begin{tikzcd}
    \mathbb{N}^{\leq n} \ar[rr]\ar[rd] & & \mathcal{B}_n \ar[ld] \\ 
     & \Vect_{\operatorname{inj}}^{\leq n}
\end{tikzcd} \ \ .
\end{equation}
In particular, this allows us to associate to a vari-framed manifold a canonical $\mathcal{B}_n$-structure on the same manifold. All the examples from Section~\ref{Sec: Tangential structures} have a natural choice for such a map. 
\end{remark}

The way topological defects form a higher category $\mathcal{D}^{\mathcal{B}^n}$ is by setting the objects to be quantum field theories. The 1-morphisms are topological defects between these theories. In the functorial framework they are functors $\Bord_n^{a\rightarrow b, \mathcal{F}}\to \mathcal{T}$ out of the stratified bordism category with two top strata and one codimension one stratum which locally looks like the plane of solutions to $e_1=0$ with the $\mathcal{B}_n$-tangential structure induced from the framing on $\R^n$. The source and target of the 1-morphism correspond to the pullback along the inclusion of $\Bord_n^a\hookrightarrow \Bord_n^{a\rightarrow b, \mathcal{F}} $ and $\Bord_n^b\hookrightarrow \Bord_n^{a\rightarrow b, \mathcal{F}}$. This is the space of local defect labels of a hyperplane in $\R^n$ equipped with its canonical $\mathcal{B}_n$-structure induced by the pointing. 

Similarly, 2-morphisms correspond to local labels for the stratification which adds the additional codimension 2-plane characterized by the equation $e_1=e_2=0$, with its $\mathcal{B}_n$-structure induced from the standart framing of $\R^n$. Again we can translate this to functors out of a bordism category locally modeled on this picture. 
This definition extends to higher morphisms by including planes of larger and larger codimension. 
The composition of morphisms corresponds to the fusion of defects, i.e. bring them close together and considering the result as one defect. The best approach, in our opinion, to make the $n$-category of topological defects rigorous is to construct an $n$-fold complete Segal space whose components consist of the spaces of labelings of planes in $\R^n$ with topological defects. For example in 2-dimensions the space $\mathcal{D}^{\mathcal{B}_2}_{2,3}$ would be labelings of 
\begin{center}
\includegraphics[scale=0.7]{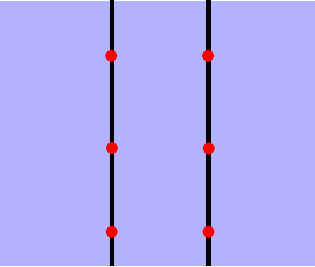}    
\end{center}
with topological defects. 
On the level of bordism categories this correspond to an $n$-fold cosimplical object in rigid symmetric monoidal $n$ categories formed from various bordism categories with defects\footnote{There is ongoing work towards constructing such an object~\cite{BN}.}.  

In the factorization algebra approach to topological quantum field theories~\cite{costello2023factorization}, this procedure has been carried out rigorously in Scheimbauer's construction of the higher Morita category~\cite{Claudiathesis}. This construction assigns to a hyperplane arrangement in $\mathbb{R}^n$ the space of constructible factorization algebras on it, which corresponds to the space of topological defects.

\begin{remark}
The reader might wonder how to reconcile the preceding discussion with the pictures in the introduction involving downwards pointing lines, which do not agree with the orientation induced from the standard framing on $\R^2$. These pictures should be interpreted as corresponding to an upwards pointing line labeled with a defect $X^\vee $ which when evaluated reverses the orientation and inserts $X$. 
\end{remark}

It is believed that $\mathcal{D}^{\mathcal{B}_n}$ has all adjoints which can be constructed by `bending' 
around topological defects. In~\cite{MSS} we will show that sheaves on vari 
framed string diagrams in $\R^n$ lead to $n$ categories with adjoints. The sketched 
construction of $\mathcal{D}^{\mathcal{B}_n}$ lifts to such a sheaf by assigning to a string diagram the 
space of consistent labelings with defects, giving a formal justification for this belief. 

We have already explained in the introduction that $\mathcal{D}^{\mathcal{B}_n}$ is equipped with additional structures depending on $\mathcal{B}_n$. The reason for this is that the category structure on $\mathcal{D}^{\mathcal{B}_n}$ corresponds to vari-framed diagrams, instead of `$\mathcal{B}_n$-diagrams'. There is an action of $O_n$ on these diagrams corresponding to rotating the diagram. Under Conjecture~\ref{Conj} this is exactly the $O_n$-action on categories with adjoints. 
In the case the class of theories only depends on a $\mathcal{B}_n$-structure instead of a vari-framing, part of this action trivializes; exactly the part corresponding to $X_n\to BO_n$. Hence $\mathcal{D}^{\mathcal{B}_n}$ should be a $X_n$-volutive category. But more is true: If we look at the space of objects of $\mathcal{D}$ corresponding to configurations with no defects inserted the action described above is trivial and the $X_n$-volution acts trivially on this space. On the space of 1-morphism the action in the orthogonal $(n-1)$-plane is trivialized and so on for spaces of higher morphisms. All these trivializations are compatible with each other. The resulting structure exactly matches the definition of an $\mathcal{B}_n$-dagger category. Hence we propose: 
\begin{proposal}
Let $\mathcal{B}_n$ be a pointed stratified $n$-dimensional tangential structure. Categories of topological defects with $\mathcal{B}_n$ structure are $\mathcal{B}_n$-dagger categories.
\end{proposal}
Some additional evidence for this proposal beyond the informal argument above is that it reproduces and unifies the few known descriptions of additional structures on categories of topological defects, in particular the connection between oriented defects in 2-dimensional quantum field theories and pivotal bicategories.     

For certain topological field theories the proposal follows from the stratified cobordism hypothesis as we explain next. 
Let $\mathcal{B}_n$ be a pointed stratified tangential structure. There is a functor from rigid $\mathcal{B}_n$-dagger categories to rigid categories equipped with an $X_n$-action
\begin{align}
(-)_n \colon \mathcal{B}_n\text{-}\dagger\RigidCat_n & \longrightarrow \RigidCat_n^{X_n} \\ 
\left( \Ca_0\rightarrow \dots \rightarrow \Ca_n \right) & \longmapsto \Ca_n 
\end{align}
which we expect to admit a right adjoint $\operatorname{Herm}\colon \RigidCat_n^{X_n} \to \mathcal{B}_n\text{-}\dagger\RigidCat_n $. 

We suggest the following explicit construction of its components at a rigid category with trivial $X_n$-action, corresponding to the canonical $X_n$-volution: 
The idea of the construction is to construct the largest choices for $\Ca_i$ which are compatible with the fully-faithfulness conditions. For this we set $\operatorname{Herm}'_{i}(\Ca)= (\iota_i\Ca)^{F_{i,n}} $ equipped with the trivial $X_i$-action, where $F_{i,j}$ is the fiber of $X_{i,j}\to BO_i$.
These categories do not from a dagger category, since in general the essential surjectivity conditions are violated. However, by restricting to objects and morphisms admitting at least one fixed point this can easily be solved. We denote the resulting $\mathcal{B}^n$-dagger category by $\operatorname{Herm}(\Ca)$ and expect this to be a model for the right adjoint. It is straightforward to verify this for simple tangential structures in low dimensions. 
\begin{example}
 For unoriented defects the fiber $F_{i,n}$ is $BO_{n-i}$ and hence $\operatorname{Herm}'(\Ca)$ is
 \begin{align}
    (\iota_0 \Ca)^{O_n} \to (\iota_1 \Ca)^{O_{n-1}} \to \dots \to (\iota_{n-1} \Ca)^{O_1}\to \Ca \ \ . 
 \end{align}
 This example generalizes to stable tangential structures with a direct sum operation such as orientations or spin structures replacing $O_i$ with $H_i$. For $n=2$ and $H=\operatorname{SO}$ this reproduces the construction in~\cite{TimNils}.    
 
We can use the stratified cobordism hypothesis to identify the spaces in $\operatorname{Herm} (\Ca)$ in this setting. The traditional cobordism hypothesis states that the space $\operatorname{Herm}(\Ca)_0= (\iota_0 \mathcal{C})^{H_n}$ where $H_n$-acts on $\iota_0 \mathcal{C}$ via the map $BH_n\to BO_n$ is the space of fully extended topological field theories with $H_n$-structure. 
The stratified version of the cobordism hypothesis~\cite[Section 4.3]{luriecobhyp} (see also~\cite[Section 2.5]{FMT}) implies that this correspondence continues and we can identify $\operatorname{Herm}_i(\Ca)=(\iota_i\Ca)^{H_{n-i}}$ with the $i$-category encoding codimension $i$-defects with $H_{n-i}$-structure.  
\end{example}
This example implies:
\begin{theorem}
Let $H$ be a stable tangential structure with direct sum. 
Assuming the stratified cobordism hypothesis the category of $H_n$-topological defects in fully extended topological quantum field theories is a rigid $H_n$-dagger category. 
\end{theorem}
We expect that the stratified cobordism hypothesis implies our proposal in general, but we were unable to verify it explicitly. One of the problems we run into is that the exact statement of the stratified cobordism hypothesis for arbitrary tangential structures is slightly elusive.    
\section{Algebraic constructions for categories of topological defects}
Many constructions with topological defects are expected to have a purely higher algebraic interpretation. To appropriately take into account different tangential structures of defects, the higher dagger structure will be crucial. We will show in two examples how the notion of higher dagger categories can be used as an organizing principle. 
\subsection{Condensation, orbifold, and unitary idempotent completions} 
As mentioned in the introduction, topological defects are often interpreted as generalized symmetries. Gauging these generalized symmetries is also referred to as condensation~\cite{GJF} or generalized orbifold constructions~\cite{ffrs0909.5013,cr1210.6363, carqueville2019orbifolds, CARQUEVILLE2025618}. 
\begin{figure}[h]
\centering
        \begin{overpic}[width=0.7\textwidth, tics=10]{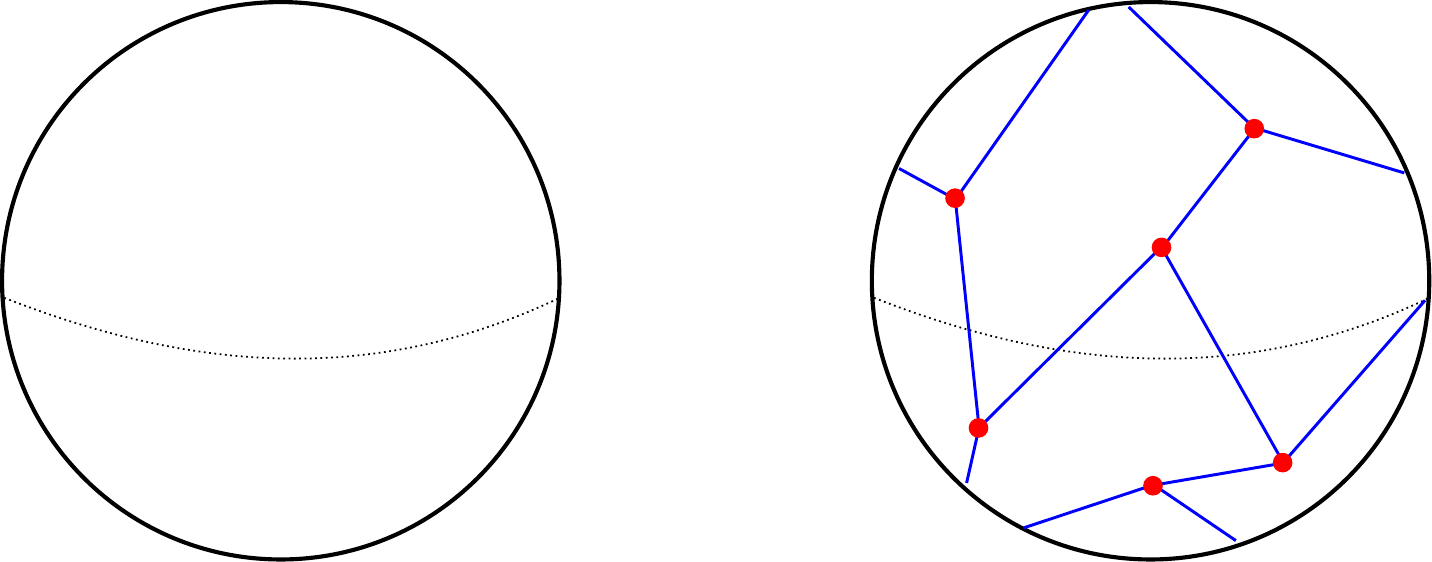}
           \put(-70,60){\huge\(\mathcal{Z}/\hspace{-0.2cm}/\mathcal{A} ( \)}
           \put(140,60){\huge \( ) \)}
           \put(150,60){\huge \(=\)}
           \put(175,60){\huge \( \mathcal{Z} ( \)}
           \put(340,60){\huge \() \)}
           \end{overpic} 
    \caption{The value of the gauge theory on a 2-dimensional manifold $M$ constructed by the evaluation of the original theory in the presence of a network of defects.}
         \label{fig:Orb}
\end{figure}
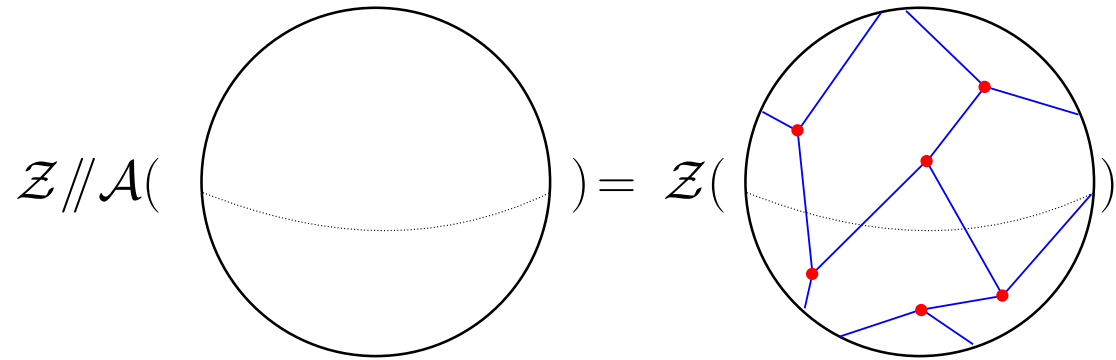

Roughly, the gauging of a generalized symmetry corresponds to picking a collection of defects $\mathcal{A}$ in a quantum field theory $\mathcal{Z}$ which can be inserted in a dense mesh into spacetimes such that the outcome of evaluating the defect field theory on it is independent of the exact from of the mesh, see Figure~\ref{fig:Orb} for a sketch. The gauged theory $\mathcal{Z}/\hspace{-0.1cm}/ \mathcal{A}$ is the theory whose partition function is the value of the original theory in the 
presents of this mesh of defects. $\mathcal{A}$ is called a \emph{condensable algebra} or \emph{condensation monad} (in the 
setting of framed theories) and a \emph{generalized orbifold datum} (in the setting of oriented theories). It should be apparent 
that the specifics of the constituents of $\mathcal{A}$ as well as the conditions they have to satisfy depend on the tangential 
structure we want the theory $\mathcal{Z}/\hspace{-0.1cm}/ \mathcal{A}$ to depend on as well as the tangential structure of the 
involved defects. 
\begin{example}\label{Ex: Orb}
For $n=2$ and oriented theories, $\mathcal A$ amounts to a $\Delta$-separable symmetric Frobenius algebra in $\End(\mathcal{Z})$ with defining relations 
\begin{equation}
\label{eq:DssFrob}
\tikzzbox{%
	\begin{tikzpicture}[very thick,scale=0.53,color=green!50!black, baseline=0.59cm]
	\draw[-dot-] (3,0) .. controls +(0,1) and +(0,1) .. (2,0);
	\draw[-dot-] (2.5,0.75) .. controls +(0,1) and +(0,1) .. (3.5,0.75);
	\draw (3.5,0.75) -- (3.5,0); 
	\draw (3,1.5) -- (3,2.25); 
	\end{tikzpicture} 
}%
=
\tikzzbox{%
	\begin{tikzpicture}[very thick,scale=0.53,color=green!50!black, baseline=0.59cm]
	\draw[-dot-] (3,0) .. controls +(0,1) and +(0,1) .. (2,0);
	\draw[-dot-] (2.5,0.75) .. controls +(0,1) and +(0,1) .. (1.5,0.75);
	\draw (1.5,0.75) -- (1.5,0); 
	\draw (2,1.5) -- (2,2.25); 
	\end{tikzpicture} 
}%
\, , \quad
\tikzzbox{%
	\begin{tikzpicture}[very thick,scale=0.33,color=green!50!black, baseline]
	\draw (-0.5,-0.5) node[Odot] (unit) {}; 
	\fill (0,0.6) circle (5.0pt) node (meet) {};
	\draw (unit) .. controls +(0,0.5) and +(-0.5,-0.5) .. (0,0.6);
	\draw (0,-1.5) -- (0,1.5); 
	\end{tikzpicture} 
}%
=
\tikzzbox{%
	\begin{tikzpicture}[very thick,scale=0.33,color=green!50!black, baseline]
	\draw (0,-1.5) -- (0,1.5); 
	\end{tikzpicture} 
}%
=
\tikzzbox{%
	\begin{tikzpicture}[very thick,scale=0.33,color=green!50!black, baseline]
	\draw (0.5,-0.5) node[Odot] (unit) {}; 
	\fill (0,0.6) circle (5.0pt) node (meet) {};
	\draw (unit) .. controls +(0,0.5) and +(0.5,-0.5) .. (0,0.6);
	\draw (0,-1.5) -- (0,1.5); 
	\end{tikzpicture} 
}%
\, , \quad
\tikzzbox{%
	\begin{tikzpicture}[very thick,scale=0.53,color=green!50!black, baseline=-0.59cm, rotate=180]
	\draw[-dot-] (3,0) .. controls +(0,1) and +(0,1) .. (2,0);
	\draw[-dot-] (2.5,0.75) .. controls +(0,1) and +(0,1) .. (1.5,0.75);
	\draw (1.5,0.75) -- (1.5,0); 
	\draw (2,1.5) -- (2,2.25); 
	\end{tikzpicture} 
}%
=
\tikzzbox{%
	\begin{tikzpicture}[very thick,scale=0.53,color=green!50!black, baseline=-0.59cm, rotate=180]
	\draw[-dot-] (3,0) .. controls +(0,1) and +(0,1) .. (2,0);
	\draw[-dot-] (2.5,0.75) .. controls +(0,1) and +(0,1) .. (3.5,0.75);
	\draw (3.5,0.75) -- (3.5,0); 
	\draw (3,1.5) -- (3,2.25); 
	\end{tikzpicture} 
}%
\, , \quad
\tikzzbox{%
	\begin{tikzpicture}[very thick,scale=0.33,color=green!50!black, baseline=0, rotate=180]
	\draw (0.5,-0.5) node[Odot] (unit) {}; 
	\fill (0,0.6) circle (5.0pt) node (meet) {};
	\draw (unit) .. controls +(0,0.5) and +(0.5,-0.5) .. (0,0.6);
	\draw (0,-1.5) -- (0,1.5); 
	\end{tikzpicture} 
}%
=
\tikzzbox{%
	\begin{tikzpicture}[very thick,scale=0.33,color=green!50!black, baseline=0, rotate=180]
	\draw (0,-1.5) -- (0,1.5); 
	\end{tikzpicture} 
}%
=
\tikzzbox{%
	\begin{tikzpicture}[very thick,scale=0.33,color=green!50!black, baseline=0cm, rotate=180]
	\draw (-0.5,-0.5) node[Odot] (unit) {}; 
	\fill (0,0.6) circle (5.0pt) node (meet) {};
	\draw (unit) .. controls +(0,0.5) and +(-0.5,-0.5) .. (0,0.6);
	\draw (0,-1.5) -- (0,1.5); 
	\end{tikzpicture} 
}%
\, , \quad 
\tikzzbox{%
	\begin{tikzpicture}[very thick,scale=0.33,color=green!50!black, baseline=0cm]
	\draw[-dot-] (0,0) .. controls +(0,-1) and +(0,-1) .. (-1,0);
	\draw[-dot-] (1,0) .. controls +(0,1) and +(0,1) .. (0,0);
	\draw (-1,0) -- (-1,1.5); 
	\draw (1,0) -- (1,-1.5); 
	\draw (0.5,0.8) -- (0.5,1.5); 
	\draw (-0.5,-0.8) -- (-0.5,-1.5); 
	\end{tikzpicture}
}%
=
\tikzzbox{%
	\begin{tikzpicture}[very thick,scale=0.33,color=green!50!black, baseline=0cm]
	\draw[-dot-] (0,1.5) .. controls +(0,-1) and +(0,-1) .. (1,1.5);
	\draw[-dot-] (0,-1.5) .. controls +(0,1) and +(0,1) .. (1,-1.5);
	\draw (0.5,-0.8) -- (0.5,0.8); 
	\end{tikzpicture}
}%
=
\tikzzbox{%
	\begin{tikzpicture}[very thick,scale=0.33,color=green!50!black, baseline=0cm]
	\draw[-dot-] (0,0) .. controls +(0,1) and +(0,1) .. (-1,0);
	\draw[-dot-] (1,0) .. controls +(0,-1) and +(0,-1) .. (0,0);
	\draw (-1,0) -- (-1,-1.5); 
	\draw (1,0) -- (1,1.5); 
	\draw (0.5,-0.8) -- (0.5,-1.5); 
	\draw (-0.5,0.8) -- (-0.5,1.5); 
	\end{tikzpicture}
}%
\, , \quad 
\begin{tikzpicture}[very thick,scale=0.33,color=green!50!black, baseline=-0.1cm]
\draw[-dot-] (0,0) .. controls +(0,-1) and +(0,-1) .. (1,0);
\draw[-dot-] (0,0) .. controls +(0,1) and +(0,1) .. (1,0);
\draw (0.5,-0.8) -- (0.5,-1.5); 
\draw (0.5,0.8) -- (0.5,1.5); 
\end{tikzpicture}
\, = \, 
\begin{tikzpicture}[very thick,scale=0.33,color=green!50!black, baseline=-0.1cm]
\draw (0.5,-1.5) -- (0.5,1.5); 
\end{tikzpicture}
\, , \quad 
\begin{tikzpicture}[very thick,scale=0.33,color=green!50!black, baseline=0cm]
\draw[-dot-] (0,0) .. controls +(0,1) and +(0,1) .. (-1,0);
\draw[directedgreen, color=green!50!black] (1,0) .. controls +(0,-1) and +(0,-1) .. (0,0);
\draw (-1,0) -- (-1,-1.5); 
\draw (1,0) -- (1,1.5); 
\draw (-0.5,1.2) node[Odot] (end) {}; 
\draw (-0.5,0.8) -- (end); 
\end{tikzpicture}
= 
\begin{tikzpicture}[very thick,scale=0.33,color=green!50!black, baseline=0cm]
\draw[redirectedgreen, color=green!50!black] (0,0) .. controls +(0,-1) and +(0,-1) .. (-1,0);
\draw[-dot-] (1,0) .. controls +(0,1) and +(0,1) .. (0,0);
\draw (-1,0) -- (-1,1.5); 
\draw (1,0) -- (1,-1.5); 
\draw (0.5,1.2) node[Odot] (end) {}; 
\draw (0.5,0.8) -- (end); 
\end{tikzpicture}
\end{equation} 
also known as an orbifold datum. Note that the last equation which is the symmetry condition can only be defined in a pivotal bicategory, because the left and right side use coevaluations of different adjunctions.  

On the other hand a condensation 2-monad~\cite{DR,GJF} which can be defined in any 2-category is a $\Delta$-separable Frobenius algebra, i.e. it satisfies the same relations as above with the exception of the last one. We refer to~\cite{FrauenbergerMasterThesis} for a detailed comparison. 
\end{example}

In~\cite{GJF} it was suggested that for framed theories $\mathcal{A}$ should algebraically be interpreted as a $n$-categorical version of an idempotent and the theory $\mathcal{Z}/\hspace{-0.1cm}/ \mathcal{A}$ provides a splitting of it. In particular, adding all theories and defects which can be constructed from a given collection of theories and defects encoded by an $n$-category $\mathcal{D}$ with adjoints, leads to a new category $\operatorname{Con}(\mathcal{D})$ the \emph{condensation completion} of $\mathcal{D}$. From an algebraic perspective this should be understood as being the idempotent or Karoubi completion of $\mathcal{D}$. 
It can be characterized by the following universal property: The $n$-category of functors from $\mathcal{D}$ to a Karoubi complete $n$-category $\mathcal{T}$ is equivalent to the category of functors from $\operatorname{Con}(\mathcal{D})$ to $\mathcal{T}$, i.e. 
\begin{equation}
\label{Eq: completion}
\begin{tikzcd}
\mathcal{D} \ar[r] \ar[d] & \mathcal{T} \\
\operatorname{Con}(\mathcal{D}) \ar[ru, dotted, "{\exists !}",swap]
\end{tikzcd}
\end{equation}

The following 1-dimensional example motivates how to extend this picture to other tangential structures.     

\begin{example}
Let $\mathcal{C}$ be a 1-category describing topological defects in a collection of oriented (or equivalently framed) 1-dimensional QFTs. A condensation monad is in this case the same as an endomorphism $\mathcal{A}\colon \mathcal{Z}\to \mathcal{Z} $ satisfying $\mathcal{A}^2=\mathcal{A}$, i.e.\ an idempotent. The value of $\mathcal{Z}/\hspace{-0.1cm}/ \mathcal{A} $ on, for example, a circle equipped with appropriate geometric structures is the evaluation of
\begin{center}
\begin{overpic}[scale=0.7]{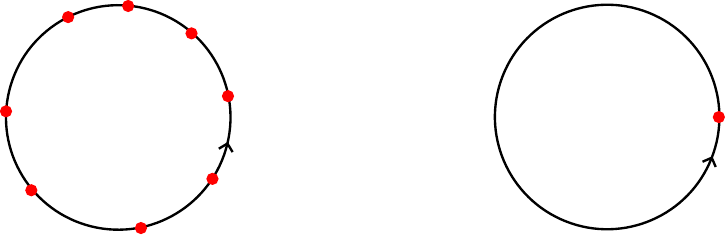} 
\put(28,6.5){{$\mathcal{Z}$}}
\put(70,12){$\mathcal{A}$} 
\put(120,35){{$=$}}
\end{overpic} \ \ .
\end{center}

Let us try to extend this construction to defects in unoriented theories, described by a dagger category $(\mathcal{C}, \dagger)$. We can choose, as above, an endomorphism $\mathcal{A}$ and define the gauged theory through the picture above. However, the resulting theory would depend on an orientation which distinguishes inserting $\mathcal{A}$ and $\mathcal{A}^\dagger$. For the construction to result in an unoriented theory, we need to impose $\mathcal{A}=\mathcal{A}^\dagger$, i.e. $\mathcal{A}$ needs to be a self-adjoint idempotent. 
We can define the splitting of a self-adjoint idempotent $\mathcal{A}$ as a 1-morphism $F\colon \mathcal{Z} \to \mathcal{Z }/\hspace{-0.1cm}/ \mathcal{A}$ such that $F^\dagger \circ F= \mathcal{A}$ and $F \circ F^\dagger = \id$. We say a dagger category is unitarily idempotent complete if every self-adjoint idempotent splits. Every dagger category admits a self adjoint idempotent completion $\operatorname{Con}^{O_1}(\Ca)$ which satisfies the analogues universal property to Diagram~\eqref{Eq: completion} where $\mathcal{T}$ is a unitarily idempotent complete dagger category and all functors involved are dagger functors. In summary gauging in unoriented 1-dimensional theories corresponds to splitting self-adjoint idempotents.    
\end{example}
This leads to the following expectation. 
\begin{expectation}
Let $\mathcal{B}_n$ be a pointed stratified tangential structure. Gauging generalized symmetries corresponding to topological defects with $\mathcal{B}_n$-structure has a categorical interpretation and implementation as splitting a $\mathcal{B}_n$-idempotent in the $\mathcal{B}_n$-dagger category $\mathcal{D}$ of topological defects. Closing $\mathcal{D}$ under gauging leads to a $\mathcal{B}_n$-dagger category $\operatorname{Con}^{\mathcal{B}_n}(\mathcal{D})$ in which every idempotent splits. Furthermore, $\operatorname{Con}^{\mathcal{B}_n}(\mathcal{D})$ satisfies an analogue's universal property to Diagram~\eqref{Eq: completion} where $\mathcal{T}$ is $\mathcal{B}_n$-idempotent complete and all functors are $\mathcal{B}_n$-dagger functors.     
\end{expectation}
We do not know how to define these objects in general, but would like to highlight a few special cases in the literature: 
In Section 4.2.2 of~\cite{CM} we study the $SO_2$ case in detail. This requires the following translation: An $SO_2$-dagger bicategory can be identified with a pivotal bicategory. What we called an $SO_2$-idempotent above is an orbifold datum, see Example~\ref{Ex: Orb} and a splitting of an orbifold datum $\mathcal{A}$ is a 1-morphism $X\colon a\to b$ such that $X^\vee \circ X \cong \mathcal{A}$ as orbifold datum (see~\cite{CM} for an explanation of why the left hand side is an orbifold datum). In Proposition 4.17 of~\cite{CM} we show that the orbifold completion of~\cite{cr1210.6363} satisfies the universal property sketched above for pivotal functors between bicategories. We suggested that this correspondence continues, i.e. that the concept of an orbifold datum is the definition of an $SO_n$-idempotent we are looking for if interpreted in higher categorical language.    

In~\cite{chen2021q} it is shown that Q-system completion satisfies a similar universal property for unitary 2-categories (we expect this result to be about $O_2$-dagger bicategories) and~\cite{chen2024manifestly} contains first steps towards $O_3$-dagger idempotent completions. In~\cite{aasen2019fermion} something closely related to the $\operatorname{Spin}_2$-case is discussed. The recent preprint~\cite{TimNils} suggests a general algebraic approach to the definition of dagger completions and gauging of non-invertible symmetries in two spacetime dimensions.       
\subsection{Euler completions}
The Euler completion introduced in~\cite{carqueville2019orbifolds} is a way of constructing new oriented defects from old ones. Algebraically it has been described as an operation for pivotal bicategories and Gray categories with duals~\cite{CM}. Let $D_k$ be a $k$-dimensional defect and $\phi$ an invertible point defect on $D_k$. We can construct a new defect $(D_k,\phi)$
whose evaluation on a stratified manifold is $D_k$ with $\phi$ to the power of the Euler characteristic of the stratum inserted. We give an interpretation of the construction from the perspective of higher dagger categories. 

Let $\mathcal{B}$ be a pivotal bicategory. The Euler completion $E(\mathcal{B})$ has as objects pairs $b\in \mathcal{B}$ together with a 2-isomorphism $\phi\colon \id_{b}\Longrightarrow \id_{b}$ and the same 1 and 2 morphisms as $\mathcal{B}$ ignoring $\phi$. Hence, as a bicategory $E(\mathcal{B})$ is equivalent to $\mathcal{B}$. However we can use the choice of $\phi$ to change the pivotal structure by changing the adjunction data for a 1-morphism $X\colon (b,\psi) \longrightarrow (b', \psi')$ as follows

\begin{equation} 
\label{eq:EulerAdjunctionMaps}
\tikzzbox{%
	\begin{tikzpicture}[very thick,scale=1.0,color=blue!50!black, baseline=0.5cm]
	\coordinate (X) at (0.5,0);
	\coordinate (Xd) at (-0.5,0);
	\coordinate (d1) at (-1,0);
	\coordinate (d2) at (+1,0);
	\coordinate (u1) at (-1,1.25);
	\coordinate (u2) at (+1,1.25);
	%
	\fill [orange!40!white, opacity=0.7] (d1) -- (d2) -- (u2) -- (u1); 
	\draw[thin] (d1) -- (d2) -- (u2) -- (u1) -- (d1); 
	%
	\draw[directed] (X) .. controls +(0,1) and +(0,1) .. (Xd);
	%
	\fill (X) circle (0pt) node[below] {{\small $X\vphantom{X^\vee}$}};
	\fill (Xd) circle (0pt) node[below] {{\small ${}^\vee\!X$}};
	\fill[red!80!black] (0,0.15) circle (0pt) node {{\scriptsize $b'$}};
	\fill[red!80!black] (0.8,1) circle (0pt) node {{\scriptsize $b$}};
	\fill[black] (0.15,0.4) circle (1.5pt) node[left] {{\scriptsize $\psi'_{b'}\hspace{-0.3em}$}};
	\fill[black] (-0.8,1) circle (1.5pt) node[right] {{\scriptsize $\hspace{-0.1em}\psi_{b}^{-1}$}};
	\end{tikzpicture}
}
\, \ \ \text{  and } \quad 
\tikzzbox{%
	\begin{tikzpicture}[very thick,scale=1.0,color=blue!50!black, baseline=0.5cm]
	\coordinate (X) at (-0.5,1.25);
	\coordinate (Xd) at (0.5,1.25);
	\coordinate (d1) at (-1,0);
	\coordinate (d2) at (+1,0);
	\coordinate (u1) at (-1,1.25);
	\coordinate (u2) at (+1,1.25);
	%
	\fill [orange!40!white, opacity=0.7] (d1) -- (d2) -- (u2) -- (u1); 
	\draw[thin] (d1) -- (d2) -- (u2) -- (u1) -- (d1);
	%
	\draw[redirected] (X) .. controls +(0,-1) and +(0,-1) .. (Xd);
	%
	\fill (X) circle (0pt) node[above] {{\small $X\vphantom{X^\vee}$}};
	\fill (Xd) circle (0pt) node[above] {{\small ${}^\vee\!X$}};
	\fill[red!80!black] (0,1.1) circle (0pt) node {{\scriptsize $b$}};
	\fill[red!80!black] (0.8,0.2) circle (0pt) node {{\scriptsize $b'$}};
	\fill[black] (0.15,0.8) circle (1.5pt) node[left] {{\scriptsize $\psi_{b}\hspace{-0.3em}$}};
	\fill[black] (-0.8,0.2) circle (1.5pt) node[right] {{\scriptsize $\hspace{-0.1em}{\psi'_{b'}}^{-1}$}};
	\end{tikzpicture}
}
\, 
\end{equation}
without changing the second adjunction. 
As explained above a pivotal bicategory is equivalent to an $SO_2$-dagger bicategory. Hence in particular $\mathcal{B}$ is equipped with an $SO_2$-volution, i.e. a trivialization $S\colon (-)^{LL}\longrightarrow \id_\mathcal{B}$ of the double left adjoint functor, as well as for every object $b\in \mathcal{B}$ a preferred trivialization of the component of this $SO_2$-volution at $b$. Now all other trivializations of $S_b$ differ from the chosen one by an automorphism $\psi \colon \id_b \Longrightarrow \id_b$. Hence the Euler completion adds all other possible $SO_2$-hermitian structures to the 2-groupoid of objects $\mathcal{B}_0$. This is exactly the hermitian completion of the underlying $SO_2$-volutive bicategory. Hence we have proven:
\begin{proposition}
Let $\mathcal{B}$ be a pivotal bicategory its Euler completion and the hermitian completion of its underlying $SO_2$-volutive category agree.  
\end{proposition}
This connection does not hold in higher dimensions, since the hermitian completion adds all possible $SO_n$-hermitian structures on objects, $SO_{n-1}$ hermitian structures on morphisms, and so on. Whereas the Euler completion only leads to mild modifications of the allowed fixed points. In physical terms, it only changes defects by stacking with an invertible phase. 
As above all these structures correspond to an $SO_i$-action, since in an $SO_n$-dagger category we have already trivialized the $SO_i$ action at the appropriate level. Constructing actions of $SO_n$ is hard in general, but for even $i$ we can look at those who are pulled back from an action of $B^{i-1}\Z$ along the Euler class $e_i \colon BSO_i \longrightarrow B^i \Z$. In an $SO_n$-dagger $(n,n)$-category adding only fixed points corresponding to actions of this type corresponds to picking an invertible $n$-endomorphism of the object or morphisms we consider. In our opinion, this construction deserves the name Euler completion and looks similar to the notion introduced in~\cite{carqueville2019orbifolds}.  We want to stress that they cannot exactly be the same because our construction does not change the objects of an $SO_3$-dagger category, but the field theoretical construction does. For any collection of characteristic classes in the right degree we can consider a similar notion of completion. For example $O_4$-dagger categories have a Pontrjagin completion.

\bibliography{biblio.bib}{}

\end{document}